\documentclass[12pt]{article}
\usepackage[height=220mm,width=150mm]{geometry}
\usepackage{array}
\usepackage{color}
\usepackage{mathrsfs}
\usepackage{hyperref}
\usepackage{amssymb}
\usepackage{amsmath}
\usepackage{amsthm}
\usepackage{color}
\allowdisplaybreaks
\usepackage{graphicx}
\usepackage{subfigure}
\usepackage{tikz}
\usetikzlibrary{shapes.arrows,chains,positioning}
\newtheorem{Theorem}{Theorem}[section]
\newtheorem{Lemma}[Theorem]{Lemma}
\newtheorem{Definition}[Theorem]{Definition}
\newtheorem{Proposition}[Theorem]{Proposition}
\newtheorem{Corollary}[Theorem]{Corollary}

\def\nocolor#1{}

\begin{document}

\title{Uncertainty Principles Associated with the Offset Linear Canonical Transform\thanks{This work was partially supported by the
National Natural Science Foundation of China (11371200, 11525104 and 11531013).}}

\author{Haiye Huo$^a$, Wenchang Sun$^b$, Li Xiao$^c$\\
$^a$ Department of Mathematics, School of Science, Nanchang University,\\ Nanchang~330031, Jiangxi, China \\
\mbox{} \\
$^b$ School of Mathematical Sciences and LPMC, Nankai University,\\ Tianjin~300071, China \\
\mbox{} \\
$^c$ Department of Biomedical Engineering, Tulane University,\\ New Orleans, LA~70118, USA\\
\mbox{} \\
\normalsize{Emails: hyhuo@ncu.edu.cn; sunwch@nankai.edu.cn; lxiao1@tulane.edu}}

\date{}
\maketitle

\textit{Abstract}.\,\,
As a time-shifted and frequency-modulated version of the linear canonical transform (LCT), the offset linear canonical transform (OLCT) provides a more general framework of most existing linear integral transforms in signal processing and optics. To study simultaneous localization of a signal and its OLCT, the classical Heisenberg's uncertainty principle has been recently generalized for the OLCT. In this paper, we complement it by presenting another two uncertainty principles, i.e., Donoho-Stark's uncertainty principle and Amrein-Berthier-Benedicks's uncertainty principle, for the OLCT. Moreover, we generalize the short-time LCT to the short-time OLCT. We likewise present Lieb's uncertainty principle for the short-time OLCT and give a lower bound for its essential support.


\section{Introduction}\label{sec:O1}
The offset linear canonical transform (OLCT) \cite{AS1994,Stern2007,XQH2013,ZWZ2016} is known as a six parameter $(a,b,c,d,\tau,\eta)$ class of linear integral transform, which is a time-shifted and frequency-modulated version of the linear canonical transform (LCT) with four parameters $(a,b,c,d)$ \cite{HS2015,SLS2012a,Stern2006A,TLW2008,XS2013}. The two extra parameters, i.e., time shifting $\tau$ and frequency modulation $\eta$, make the OLCT more general and flexible, and thereby the OLCT can apply to most electrical and optical signal systems. It basically says that the Fourier transform (FT), the fractional Fourier transform (FrFT), the Fresnel transform (FnT), the LCT, and many other widely used linear integral transforms in signal processing and optics are all special cases of the OLCT \cite{Almeida1994,LTW2007,OZK2001,Stern2006B,WL2015}. Therefore, it is interesting to study the OLCT in a unified viewpoint of the above mentioned transforms. Over the past few decades, there has been a vast amount of research on extending time-frequency analysis results that pertain to the FT or the LCT to the OLCT. For example, sampling theorems, convolution and correlation theorems, eigenfunctions, energy concentration problems, generalized prolate spheroidal wave functions, and spectral analysis of the OLCT have been derived in \cite{Stern2007,XQ2014,PD2003,Zayed2015,KMZ2013,XCH2015,XCHJL2015,Zhang2016}.

Recently, the classical Heisenberg's uncertainty principle has been generalized for the OLCT in \cite{Stern2007}.  It is stated in \cite{Stern2007} that a nonzero function and its OLCT cannot both be sharply concentrated and the joint concentration is dependent on the OLCT parameter $b$. Consider that in terms of the meaning of ``concentration'', many different forms of uncertainty principles are possible, like Heisenberg's uncertainty principle, Hardy's uncertainty principle, Donoho-Stark's uncertainty principle, and Amrein-Berthier-Benedicks's uncertainty principle \cite{FGS1997,Grochenig2001,HJ2012,SLZ2014,CSS2000,MOP2004,TY2010,AB1977,Benedicks1985}. In this paper, we derive another two common uncertainty principles for the OLCT, i.e., Donoho-Stark's uncertainty principle and Amrein-Berthier-Benedicks's uncertainty principle, which are analogous to the ones for the FT in \cite{Grochenig2001}. Moreover, we introduce the short-time OLCT as a generalization of the short-time LCT \cite{KXZ2012,KX2012}. Also, we derive Lieb's uncertainty principle for the short-time OLCT and give a lower bound for its essential support in this paper.

The rest of the paper is organized as follows. In section \ref{sec:OADD}, we briefly review the OLCT and its basic properties. In section \ref{sec:O2}, Donoho-Stark's uncertainty principle and Amrein-Berthier-Benedicks's uncertainty principle are extended for the OLCT. In section \ref{sec:O3}, we study the short-time OLCT and its Lieb's uncertainty principle. In section \ref{sec:O4}, we conclude the paper.

\section{Offset Linear Canonical Transform}\label{sec:OADD}

In this section, we briefly review the OLCT and some of its properties.

\begin{Definition}\label{def:OLCT}
The OLCT of a function $f(t)\in L^2(\mathbb{R})$ with parameter
$A=\left[
      \begin{array}{cc|c}
        a & b & \tau \\
        c & d & \eta \\
      \end{array}
      \right]$
is defined by \cite{XCH2015,XCHJL2015,BZ2017,PD2007}
\begin{equation}\label{eq:OLCT}
\mathcal{O}_{A}f(u)=\mathcal{O}_{A}[f(t)](u)=
    \begin{cases}
    \int_{-\infty}^{+\infty}f(t)K_{A}(t,u){\rm{d}}t, & b\ne0,\\
     \sqrt{d}e^{j\frac{cd}{2}(u-\tau)^2+ju\eta}f(d(u-\tau)), & b=0,\\
     \end{cases}
\end{equation}
where
\begin{equation}\label{eq:OLCT:1}
K_{A}(t,u)=\sqrt{\frac{1}{j2\pi b}}e^{j\frac{a}{2b}t^2-j\frac{1}{b}t(u-\tau)-j\frac{1}{b}u(d\tau-b\eta)+j\frac{d}{2b}(u^2+\tau^2)},
\end{equation}
$a,\;b,\;c,\;d,\;\tau,\;\eta\in\mathbb{R}$, and $ad-bc=1$.
\end{Definition}

It is readily verified from Definition~\ref{def:OLCT} that many well-known linear integral transforms are special cases of the OLCT in (\ref{eq:OLCT}). For example, when $A=\left[
      \begin{array}{cc|c}
        0 & 1 & 0 \\
        -1 & 0 & 0 \\
      \end{array}
      \right]$,
the OLCT becomes the FT; when
$A=\left[
      \begin{array}{cc|c}
        \cos\alpha & \sin\alpha & 0 \\
        -\sin\alpha & \cos\alpha & 0 \\
      \end{array}
      \right],$
the OLCT becomes the FrFT; when
$A=\left[
      \begin{array}{cc|c}
        a & b & 0 \\
        c & d & 0 \\
      \end{array}
      \right],$
the OLCT becomes the LCT.

Without loss of generality, we only consider the case of $b>0$ throughout the paper. Many properties of the OLCT have been easily verified by using the definition (\ref{eq:OLCT}) in \cite{XQH2013,PD2003,BZ2017}. Here, we present some important properties for later use.

\begin{enumerate}
  \item[$1)$] Additivity Property: Let
  \[
  A_1=\left[
      \begin{array}{cc|c}
        a_1 & b_1 & \tau_1 \\
        c_1 & d_1 & \eta_1 \\
      \end{array}
      \right],
  A_2=\left[
      \begin{array}{cc|c}
        a_2 & b_2 & \tau_2 \\
        c_2 & d_2 & \eta_2 \\
      \end{array}
      \right],
  A_3=\left[
      \begin{array}{cc|c}
        a_3 & b_3 & \tau_3 \\
        c_3 & d_3 & \eta_3 \\
      \end{array}
      \right],
  \]
  satisfy
  \begin{equation}
    \left[
      \begin{array}{cc}
        a_3 & b_3 \\
        c_3 & d_3 \\
      \end{array}
    \right]
    =\left[
      \begin{array}{cc}
        a_2 & b_2 \\
        c_2 & d_2 \\
      \end{array}
    \right]
    \left[\begin{array}{cc}
        a_1 & b_1 \\
        c_1 & d_1 \\
      \end{array}
      \right],\label{eq:A1}
      \end{equation}
      \begin{equation}
      \left(
        \begin{array}{c}
            \tau_3\\
            \eta_3 \\
        \end{array}
        \right)
        =\left[
      \begin{array}{cc}
        a_2 & b_2 \\
        c_2 & d_2 \\
      \end{array}
      \right]
      \left(
             \begin{array}{c}
               \tau_1\\
               \eta_1\\
             \end{array}
           \right)
           +\left(
             \begin{array}{c}
               \tau_2\\
               \eta_2 \\
             \end{array}
           \right).
      \end{equation}
  Then, we have
  \begin{equation}\label{OLCT:pro1}
  \mathcal{O}_{A_2}[\mathcal{O}_{A_1}f(v)](u)=e^{j\phi}\mathcal{O}_{A_3}f(u),
  \end{equation}
  where
  \[
  \phi=-\frac{a_2c_2}{2}\tau_1^2-b_2c_2\tau_1\eta_1-\frac{b_2d_2}{2}\eta_1^2-(\tau_1c_2+\eta_1d_2)\tau_2.
  \]
  \item[$2)$] Inverse OLCT: According to the above additivity property, the inverse OLCT of $\mathcal{O}_{A}f(u)$ with parameter
  $A=\left[
      \begin{array}{cc|c}
        a & b & \tau \\
        c & d & \eta \\
      \end{array}
      \right]$
  is given by
  \begin{eqnarray}
  f(t)=e^{j\frac{cd}{2}\tau^2-jad\tau\eta+j\frac{ab}{2}\eta^2}\int_{-\infty}^{+\infty}\mathcal{O}_{A}f(u)K_{A^{-1}}(u,t){\rm{d}}u,\label{eq:R-OLCT}
  \end{eqnarray}
  where
  \[
  A^{-1}=\left[
      \begin{array}{cc|c}
        d & -b & b\eta-d\tau \\
        -c & a & c\tau-a\eta \\
      \end{array}
      \right].
  \]
  \item[$3)$] Generalized Parseval Formula:
  \begin{equation}\label{OLCT:pro3}
  \int_{\mathbb{R}}f(t)\overline{g(t)}\textrm{d}t=\int_{\mathbb{R}}\mathcal{O}_{A}f(u)\overline{\mathcal{O}_{A}g(u)}\textrm{d}u,
  \end{equation}
  where $\bar{\cdot}$ denotes the complex conjugate.
\end{enumerate}

\section{Two Uncertainty Principles for the OLCT}\label{sec:O2}

Uncertainty principle for the FT plays a vital role in time-frequency signal analysis. It can express limitations on simultaneous concentration of a function and its FT. There exist many different forms of uncertainty principles, like Heisenberg's uncertainty principle, Hardy's uncertainty principle, Donoho-Stark's uncertainty principle, and Amrein-Berthier-Benedicks's uncertainty principle, in terms of the meaning of ``concentration''. To investigate simultaneous concentration of a function and its OLCT, it is natural to extend the above uncertainty principles for the OLCT. In \cite{Stern2007}, the classical Heisenberg's uncertainty principle for the OLCT has been proposed.
In this section, therefore, we complement it by considering another two common uncertainty principles, i.e., Donoho-Stark's uncertainty principle and Amrein-Berthier-Benedicks's uncertainty principle, for the OLCT.

Before presenting our results, let us define the FT of a function $f(t)\in L^2(\mathbb{R})$ by
$\mathcal{F}f(u)=\frac{1}{\sqrt{2\pi}}\int_{-\infty}^{+\infty}f(t)e^{-jut}{\rm{d}}t$ in this paper, and
recall the concept of $\epsilon$-concentrate of a function on a measurable set $\Omega\subseteq \mathbb{R}$, Donoho-Stark's uncertainty principle \cite{Grochenig2001,DS1989} and Amrein-Berthier-Benedicks's uncertainty principle for the FT \cite{AB1977,Benedicks1985,Grochenig2001} as follows.

\begin{Definition}\cite[Defintion~2.3.1]{Grochenig2001}\label{Def:OLCT:O1}
Given $\epsilon\ge 0$, a function $f(t)\in L^2(\mathbb{R})$ is $\epsilon$-concentrate on a measurable set $\Omega\subseteq \mathbb{R}$, if
\begin{equation}\label{Def:OLCT:O2}
\Big(\int_{\mathbb{R}\backslash\Omega}|f(t)|^2{\rm{d}}t\Big)^{1/2}\le\epsilon\|f\|_2.
\end{equation}
\end{Definition}

\begin{Proposition}\cite[Theorem~2.3.1 (Donoho-Stark)]{Grochenig2001}\label{Pro:OLCT:UP3}
Suppose that a nonzero function $f(t)\in L^2(\mathbb{R})$ is $\epsilon_{\Omega}$-concentrate on a measurable set $\Omega\subseteq \mathbb{R}$, and its FT $\mathcal{F}f(u)$ is $\epsilon_{\Gamma}$-concentrate on a measurable set $\Gamma\subseteq \mathbb{R}$. Then,
\begin{equation}\label{Pro:OLCT:UP4}
|\Omega||\Gamma|\ge2\pi(1-\epsilon_{\Omega}-\epsilon_{\Gamma})^2,
\end{equation}
where $|\Omega|$ and $|\Gamma|$ denote the measures of the sets $\Omega$ and $\Gamma$.
\end{Proposition}

\begin{Proposition}\cite[Theorem~2.3.3 (Amrein-Berthier-Benedicks)]{Grochenig2001}\label{Pro:OLCT:C1}
Let $f(t)\in L^1(\mathbb{R})$,\;${\rm{supp}}(f)\subseteq\Omega$, and
${\rm{supp}}(\mathcal{F}f)\subseteq\Gamma,$ where $\Omega$ and $\Gamma$ are measurable sets in $\mathbb{R}$. If $|\Omega||\Gamma|<+\infty$, then $f(t)=0$.
\end{Proposition}

Motivated by Propositions \ref{Pro:OLCT:UP3}, \ref{Pro:OLCT:C1}, we next generalize the corresponding results for the OLCT. We first present Donoho-Stark's uncertainty principle for the OLCT in the following theorem.

\begin{Theorem}\label{Thm:OLCT:A1}
Suppose that a nonzero function $f(t)\in L^2(\mathbb{R})$ is $\epsilon_{\Omega}$-concentrate on a measurable set $\Omega\subseteq \mathbb{R}$, and its OLCT $\mathcal{O}_{A}f(u)$ is $\epsilon_{\Gamma}$-concentrate on a measurable set $\Gamma\subseteq \mathbb{R}$. Then,
\begin{equation}\label{Thm:OLCT:A2}
|\Omega||\Gamma|\ge2\pi b(1-\epsilon_{\Omega}-\epsilon_{\Gamma})^2.
\end{equation}
\end{Theorem}

\begin{proof}
From Definition \ref{def:OLCT}, the OLCT with parameter $A$ can be equivalently written as
\begin{equation}\label{Thm:OLCT:A4}
\mathcal{O}_{A}f(u)=\sqrt{\frac{1}{jb}}e^{-j\frac{1}{b}u(d\tau-b\eta)+j\frac{d}{2b}(u^2+\tau^2)}G(u),
\end{equation}
where
\begin{equation}\label{Thm:OLCT:A3}
G(u)=\frac{1}{\sqrt{2\pi}}\int_{-\infty}^{+\infty}f(t)e^{j\frac{a}{2b}t^2-j\frac{1}{b}t(u-\tau)}{\rm{d}}t.
\end{equation}
In this way, it is easy to see that
\begin{equation}\label{Thm:OLCT:A5}
|\mathcal{O}_{A}f(u)|=\sqrt{\frac{1}{b}}|G(u)|.
\end{equation}
Since $\mathcal{O}_{A}f(u)$ is $\epsilon_{\Gamma}$-concentrate on a measurable set $\Gamma\subseteq \mathbb{R}$, we have
\begin{equation}\label{Thm:OLCT:A6}
\Big(\int_{\mathbb{R}\backslash\Gamma}|\mathcal{O}_{A}f(u)|^2{\rm{d}}u\Big)^{1/2}\le\epsilon_{\Gamma}\|\mathcal{O}_{A}f\|_2.
\end{equation}
Then, substituting (\ref{Thm:OLCT:A5}) into (\ref{Thm:OLCT:A6}), we have
\begin{equation}\label{Thm:OLCT:A7}
\Big(\int_{\mathbb{R}\backslash\Gamma}|G(u)|^2{\rm{d}}u\Big)^{1/2}\le\epsilon_{\Gamma}\|G\|_2,
\end{equation}
Thus, $G(ub)$ is $\epsilon_{\Gamma}$-concentrate on a measurable set $\Gamma/b\subseteq \mathbb{R}$. Furthermore, it is shown in (\ref{Thm:OLCT:A3}) that $G(ub)$ is actually the FT of $g(t)=f(t)e^{j\frac{a}{2b}t^2+j\frac{1}{b}t\tau}$.
Moreover, since $|g(t)|=|f(t)|$ and $f(t)$ is $\epsilon_{\Omega}$-concentrate on a measurable set $\Omega\subseteq \mathbb{R}$, i.e.,
\begin{equation}\label{Thm:OLCT:A8}
\Big(\int_{\mathbb{R}\backslash\Omega}|f(t)|^2{\rm{d}}t\Big)^{1/2}\le\epsilon_{\Omega}\|f\|_2,
\end{equation}
we obtain
\begin{equation}\label{Thm:OLCT:A9}
\Big(\int_{\mathbb{R}\backslash\Omega}|g(t)|^2{\rm{d}}t\Big)^{1/2}\le\epsilon_{\Omega}\|g\|_2.
\end{equation}
In other words, $g(t)$ is $\epsilon_{\Omega}$-concentrate on a measurable set $\Omega\subseteq \mathbb{R}$. As for the function $g(t)$ and its FT $G(ub)$, we have, by Proposition~\ref{Pro:OLCT:UP3},
\begin{equation}\label{Thm:OLCT:A10}
|\Omega|\Big|\frac{\Gamma}{b}\Big|\ge2\pi(1-\epsilon_{\Omega}-\epsilon_{\Gamma})^2.
\end{equation}
Therefore,
\[
|\Omega||\Gamma|\ge2\pi b(1-\epsilon_{\Omega}-\epsilon_{\Gamma})^2.
\]
This completes the proof.
\end{proof}

\begin{Corollary}\label{Coro:OLCT:B}
Let $f(t)\in L^2(\mathbb{R})$,\;${\rm{supp}}(f)\subseteq\Omega$, and ${\rm{supp}}(\mathcal{O}_{A}f)\subseteq\Gamma$, where $\Omega$ and $\Gamma$ are measurable sets in $\mathbb{R}$. Then,
\[
|\Omega||\Gamma|\ge2\pi b.
\]
\end{Corollary}

\begin{proof}
From Definition~\ref{Def:OLCT:O1}, we know that a function $f(t)$ is 0-concentrate on a measurable set $\Omega\subseteq \mathbb{R}$ if and only if ${\rm{supp}}(f)=\Omega.$ Hence, by letting $\epsilon_{\Omega}=\epsilon_{\Gamma}=0$ in Theorem~\ref{Thm:OLCT:A1}, we get $|\Omega||\Gamma|\ge2\pi b$.
\end{proof}

We next present Amrein-Berthier-Benedicks's uncertainty principle for the OLCT in the following theorem.

\begin{Theorem}\label{Thm:OLCT:D1}
Let $f(t)\in L^1(\mathbb{R})$, ${\rm{supp}}(f)\subseteq\Omega$, and ${\rm{supp}}(\mathcal{O}_{A}f)\subseteq\Gamma,$ where $\Omega$ and $\Gamma$ are measurable sets in $\mathbb{R}$. If $|\Omega||\Gamma|<+\infty$, then
$f(t)=0.$
\end{Theorem}

\begin{proof}
Similar to the proof of Theorem \ref{Thm:OLCT:A1}, let the OLCT $\mathcal{O}_{A}f(u)$ be rewritten in the form of (\ref{Thm:OLCT:A4}), where $G(u)$ is given by (\ref{Thm:OLCT:A3}). Since ${\rm{supp}}(\mathcal{O}_{A}f)\subseteq\Gamma$, we easily get ${\rm{supp}}(G(u))\subseteq\Gamma$. Moreover, by letting $g(t)=f(t)e^{j\frac{a}{2b}t^2+j\frac{1}{b}t\tau}$, we have $g(t)\in L^1(\mathbb{R})$ and ${\rm{supp}}(g)\subseteq\Omega$. This is due to the fact that $|f(t)|=|g(t)|$, $f(t)\in L^1(\mathbb{R})$, and ${\rm{supp}}(f)\subseteq\Omega$. Considering that ${\rm{supp}}(G(ub))\subseteq\Gamma/b$ and $G(ub)$ is the FT of $g(t)$, we obtain $g(t)=0$ by Proposition~\ref{Pro:OLCT:C1}. Hence, we get $f(t)=0$.
\end{proof}

One can see that Theorem \ref{Thm:OLCT:A1} and Theorem \ref{Thm:OLCT:D1} for the OLCT with the specific parameter
$A=\left[
      \begin{array}{cc|c}
        0 & 1 & 0 \\
        -1 & 0 & 0 \\
      \end{array}
      \right]$
(i.e., the FT), coincide with Proposition \ref{Pro:OLCT:UP3} and Proposition \ref{Pro:OLCT:C1}, respectively. Furthermore, Donoho-Stark's uncertainty principle and Amrein-Berthier-Benedicks's uncertainty principle for the LCT can also be directly obtained by substituting
$A=\left[
      \begin{array}{cc|c}
        a & b & 0 \\
        c & d & 0 \\
      \end{array}
      \right]$
for the OLCT in Theorem \ref{Thm:OLCT:A1} and Theorem  \ref{Thm:OLCT:D1}.

\section{Lieb's Uncertainty Principle for the Short-Time OLCT}\label{sec:O3}

Similar to the FT or the LCT, the OLCT cannot reveal the local OLCT-frequency information due to its global kernel. In this section, therefore, we introduce the short-time OLCT, which generalizes the short-time LCT by substituting the LCT kernel with the OLCT kernel in the definition of the short-time LCT \cite{KXZ2012,KX2012}. We then generalize Lieb's uncertainty principle for the short-time LCT in \cite{KXZ2012} to our proposed short-time OLCT.

\begin{Definition}\label{def:WOLCT}
Given a function $g(t)\in L^{\infty}(\mathbb{R})$, the short-time OLCT of a function $f(t)\in L^1(\mathbb{R})$ with a window $g(t)$ is defined by
\begin{equation}\label{eq:WOLCT}
V_{g,A}f(x,u)=\int_{-\infty}^{+\infty}f(t)\overline{g(t-x)}K_{A}(t,u){\rm{d}}t, \quad \textrm{for}\;\;x,\;u\in \mathbb{R},
\end{equation}
where $K_{A}(t,u)$ is given by (\ref{eq:OLCT:1}),\;$a,\;b,\;c,\;d,\;\tau,\;\eta\in\mathbb{R}$,\;$b>0$, and $ad-bc=1$.
\end{Definition}

Based on H\"{o}lder's inequality, the short-time OLCT $V_{g,A}f(x,u)$ in (\ref{eq:WOLCT}) of a function $f(t)\in L^p(\mathbb{R})$ with a window $g(t)\in L^q(\mathbb{R})$ is well-defined for any $p,\;q\in[1,+\infty]$ satisfying $1/p+1/q=1$. On the other hand, from the definition of the short-time OLCT $V_{g,A}f$ in (\ref{eq:WOLCT}), the relation between the short-time OLCT and the FT is given by
\begin{equation}\label{Rem:WOLCT:D1}
V_{g,A}f(x,bu+\tau)=\sqrt{\frac{1}{jb}}\mathcal{F}\big[f(\cdot)\overline{g(\cdot-x)}e^{j\frac{a}{2b}(\cdot)^2}\big]\big(u\big)
          e^{j\frac{d}{2b}(bu)^2+j(bu+\tau)\eta},
\end{equation}
where $\mathcal{F}$ denotes the FT operator. By (\ref{Rem:WOLCT:D1}), we have, for $f(t),\;g(t)\in L^2(\mathbb{R})$,
\begin{equation}\label{Lem:WOLCT:A1}
\|V_{g,A}f\|_2=\|f\|_2\|g\|_2.
\end{equation}
.

Before presenting Lieb's uncertainty principle for the short-time OLCT, we show a generalized Hausdorff-Young inequality in the following lemma, whose proof is similar to that of Hausdorff-Young inequality in \cite{GXX2009}.

\begin{Lemma}\label{Lem:WOLCT:UP}
Let
$A_1=\left[
      \begin{array}{cc|c}
        a_1 & b_1 & \tau_1 \\
        c_1 & d_1 & \eta_1 \\
      \end{array}
      \right]$,
$A_2=\left[
      \begin{array}{cc|c}
        a_2 & b_2 & \tau_2 \\
        c_2 & d_2 & \eta_2 \\
      \end{array}
      \right]$,
and $a_2b_1-a_1b_2>0.$ Let $1\le q\le 2$ and $p$ satisfy $1/p+1/q=1$. Then, for all $f(t)\in L^{q}(\mathbb{R})$, we have
\begin{equation}\label{Lem:WOLCT:UP1}
\|\mathcal{O}_{A_1}f\|_{p}\le \Big((2\pi)^{1/p-1/q}q^{1/q}p^{-1/p}\Big)^{1/2}(a_2b_1-a_1b_2)^{1/2-1/q}\|\mathcal{O}_{A_2}f\|_{q}.
\end{equation}
\end{Lemma}

\begin{proof}
Let
\[
A_3=\left[
      \begin{array}{cc|c}
        a_3 & b_3 & \tau_3 \\
        c_3 & d_3 & \eta_3 \\
      \end{array}
      \right],
\]
which satisfies
\begin{equation}
\left(
      \begin{array}{cc}
        a_2 & b_2 \\
        c_2 & d_2 \\
      \end{array}
    \right)
    =\left(
      \begin{array}{cc}
        d_3 & -b_3 \\
        -c_3 & a_3 \\
      \end{array}
    \right)
    \left(\begin{array}{cc}
        a_1 & b_1 \\
        c_1 & d_1 \\
      \end{array}
      \right),\label{eq:add1}
      \end{equation}
      and
      \begin{equation}
\left(
        \begin{array}{c}
            \tau_2\\
            \eta_2 \\
        \end{array}
        \right)
        =\left(
      \begin{array}{cc}
        d_3 & -b_3 \\
        -c_3 & a_3 \\
      \end{array}
      \right)
      \left(
             \begin{array}{c}
               \tau_1\\
               \eta_1\\
             \end{array}
           \right)
           +\left(
             \begin{array}{c}
               \tau_3\\
               \eta_3 \\
             \end{array}
           \right).
\end{equation}
It follows from (\ref{eq:add1}) that $b_3=a_2b_1-a_1b_2.$
Let
\begin{equation}
H(u)=\mathcal{O}_{A_1}f(u)e^{-j\frac{d_3}{2b_3}u^2}\label{Lem:WOLCT:UP2}
\end{equation}
and its inverse FT be given by
\begin{equation}
h(t)=\frac{1}{\sqrt{2\pi}}\int_{-\infty}^{+\infty}H(u)e^{jut}{\rm{d}}u.\label{Lem:WOLCT:UP3}
\end{equation}
From (\ref{Lem:WOLCT:UP3}) and Hausdorff-Young's inequality in \cite{Grochenig2001,GXX2009}, we obtain
\begin{equation}\label{Lem:WOLCT:UP5}
\|\mathcal{O}_{A_1}f\|_{p}=\|H\|_{p}\le \Big((2\pi)^{1/p-1/q}q^{1/q}p^{-1/p}\Big)^{1/2}\|h\|_{q}.
\end{equation}
Moreover, it follows from (\ref{OLCT:pro1}), (\ref{Lem:WOLCT:UP2}) and (\ref{Lem:WOLCT:UP3}) that
\begin{small}
\begin{align}
\|h\|_{q}^q=&\int_{-\infty}^{+\infty}|h(t)|^q{\rm{d}}t\nonumber\\
=&\frac{1}{b_3}\int_{-\infty}^{+\infty}\big|h\big(\frac{t-\tau_3}{b_3}\big)\big|^q{\rm{d}}t\nonumber\\
=&\frac{1}{b_3}\int_{-\infty}^{+\infty}\Big|\frac{1}{\sqrt{2\pi}}\int_{-\infty}^{+\infty}H(u)
        e^{ju\big(\frac{t-\tau_3}{b_3}\big)}{\rm{d}}u\Big|^q{\rm{d}}t\nonumber\\
=&\frac{1}{b_3}\int_{-\infty}^{+\infty}\Big|\frac{1}{\sqrt{2\pi}}\int_{-\infty}^{+\infty}\mathcal{O}_{A_1}f(u)e^{-j\frac{d_3}{2b_3}u^2
    +ju\big(\frac{t-\tau_3}{b_3}\big)}{\rm{d}}u\Big|^q{\rm{d}}t\nonumber\\
=&\frac{1}{b_3}\int_{-\infty}^{+\infty}\bigg|\frac{\sqrt{\frac{1}{-j2\pi b_3}}\int_{-\infty}^{+\infty}\mathcal{O}_{A_1}f(u)e^{-j\frac{d_3}{2b_3}u^2+
          ju\big(\frac{t-\tau_3}{b_3}\big)-j\frac{a_3}{2b_3}(t-\tau_3)^2+j\eta_3 t}{\rm{d}}u}
           {\sqrt{\frac{1}{-jb_3}}e^{-j\frac{a_3}{2b_3}(t-\tau_3)^2+j\eta_3 t}}\bigg|^q{\rm{d}}t\nonumber\\
=&b_3^{q/2-1}\int_{-\infty}^{+\infty}\big|\mathcal{O}_{A_3}[\mathcal{O}_{A_1}f(u)](t)\big|^q{\rm{d}}t\nonumber\\
=&b_3^{q/2-1}\|\mathcal{O}_{A_2}f\|_q^q.\label{Lem:WOLCT:UP6}
\end{align}
\end{small}
Combining (\ref{Lem:WOLCT:UP5}) and (\ref{Lem:WOLCT:UP6}), we have
\begin{eqnarray*}
\|\mathcal{O}_{A_1}f\|_p
&\le&\Big((2\pi)^{1/p-1/q}q^{1/q}p^{-1/p}\Big)^{1/2}\|h\|_{q}\\
&\le&\Big((2\pi)^{1/p-1/q}q^{1/q}p^{-1/p}\Big)^{1/2}(a_2b_1-a_1b_2)^{1/2-1/q}\|\mathcal{O}_{A_2}f\|_q,
\end{eqnarray*}
which completes the proof.
\end{proof}

We next present Lieb's uncertainty principle for the short-time OLCT as follows.

\begin{Theorem}\label{Thm:WOLCT:Lieb}
Let $f(t),\;g(t)\in L^2(\mathbb{R})$ and $2\le p<\infty$. Then,
\begin{equation}\label{Thm:WOLCT:Lieb1}
\int\int_{\mathbb{R}^2}|V_{g,A}f(x,u)|^p{\rm{d}}x{\rm{d}}u\le \frac{2}{p}\big(2\pi b\big)^{1-p/2}\|f\|_2^p\|g\|_2^p.
\end{equation}
\end{Theorem}

\begin{proof}
Let $q$ satisfy $1/p+1/q=1$. Due to $2\le p<\infty$, we have $1<q\le 2$. Since $f(t),\;g(t)$ belong to $L^2(\mathbb{R}),$ we obtain
$f(\cdot)\overline{g(\cdot-x)}\in L^1(\mathbb{R})$
by Cauchy-Schwartz inequality and $V_{g,A}f(x,u)=\mathcal{O}_{A}(f(t)\overline{g(t-x)})(u)\in L^2(\mathbb{R}^2)$ by (\ref{Lem:WOLCT:A1}). Moreover, according to Fubini's Theorem, for almost all $x\in \mathbb{R}$, we get $\mathcal{O}_{A}(f(t)\overline{g(t-x)})(u)\in L^2(\mathbb{R})$, and thereby $f(\cdot)\overline{g(\cdot-x)}\in L^2(\mathbb{R})$ holds for almost all $x\in \mathbb{R}$ by the generalized Parseval Formula in (\ref{OLCT:pro3}). Therefore, we have $f(\cdot)\overline{g(\cdot-x)}\in L^1(\mathbb{R})\cap L^2(\mathbb{R})$ for almost all $x\in \mathbb{R}$. Furthermore, for almost all $x\in \mathbb{R}$, we obtain $f(\cdot)\overline{g(\cdot-x)}\in L^q(\mathbb{R})$ by using interpolation theorem. Then, by Lemma~\ref{Lem:WOLCT:UP}, we have
\begin{align}
\Big(\int_{\mathbb{R}}|V_{g,A}f(x,u)|^p{\rm{d}}u\Big)^{1/p}=&\|\mathcal{O}_{A}(f(\cdot)\overline{g(\cdot-x)})\|_{p}\nonumber\\
\le&\Big((2\pi)^{1/p-1/q}q^{1/q}p^{-1/p}\Big)^{1/2}(a_2b-ab_2)^{1/2-1/q}\nonumber\\
     &\quad\times\|\mathcal{O}_{A_2}(f(\cdot)\overline{g(\cdot-x)})\|_{q}.\label{Thm:WOLCT:Lieb2}
\end{align}
Specially, let
\[
A_2=\left[
      \begin{array}{cc|c}
        1 & 0 & 0 \\
        0 & 1 & 0 \\
      \end{array}
      \right]
\]
in (\ref{Thm:WOLCT:Lieb2}), and
we have
\begin{align}
\Big(\int_{\mathbb{R}}|V_{g,A}f(x,u)|^p{\rm{d}}u\Big)^{1/p}\le&\Big((2\pi)^{1/p-1/q}q^{1/q}p^{-1/p}\Big)^{1/2}b^{1/2-1/q}\nonumber\\
     &\quad\times\Big(\int_{\mathbb{R}}|f(t)\overline{g(t-x)}|^q{\rm{d}}t\Big)^{1/q}.\label{Thm:WOLCT:Lieb3}
\end{align}
Furthermore,
\begin{align}
\|V_{g,A}f\|_{p}=&\Big(\int_{\mathbb{R}}\Big(\int_{\mathbb{R}}|V_{g,A}f(x,u)|^p{\rm{d}}u\Big){\rm{d}}x\Big)^{1/p}\nonumber\\
\le&\Big((2\pi)^{1/p-1/q}q^{1/q}p^{-1/p}\Big)^{1/2}b^{1/2-1/q}\nonumber\\
    &\quad \times \Big(\int_{\mathbb{R}}\Big(\int_{\mathbb{R}}|f(t)\overline{g(t-x)}|^q{\rm{d}}t\Big)^{p/q}{\rm{d}}x\Big)^{1/p}\nonumber\\
\le&(2\pi)^{1/(2p)-1/(2q)}b^{1/2-1/q}\Big(\frac{2}{p}\Big)^{1/p}\|f\|_2\|g\|_2.\label{Thm:WOLCT:Lieb4}
\end{align}
where we use Young's inequality in the last step.
Therefore, we have
\begin{eqnarray*}
\int\int_{\mathbb{R}^2}|V_{g,A}f(x,u)|^p{\rm{d}}u{\rm{d}}x
&\le&\frac{2}{p}(2\pi)^{1/2-p/(2q)}b^{p/2-p/q}\|f\|_2^p\|g\|_2^p\\
&=&\frac{2}{p}\big(2\pi b\big)^{1-p/2}\|f\|_2^p\|g\|_2^p.
\end{eqnarray*}
The proof is completed.
\end{proof}

Finally, according to the Lieb's uncertainty principle for short-time OLCT in Theorem \ref{Thm:WOLCT:Lieb}, we derive a lower bound for the essential support of the short-time OLCT in the following result.

\begin{Corollary}\label{Thm:WOLCT:h1}
Let $f(t),\;g(t)\in L^2(\mathbb{R})$ and $\|f\|_2=\|g\|_2=1.$ If
\begin{equation}\label{Thm:WOLCT:h2}
\int\int_{\Omega}|V_{g,A}f(x,u)|^2{\rm{d}}x{\rm{d}}u\ge 1-\epsilon
\end{equation}
holds on a measurable set $\Omega\subseteq \mathbb{R}^2$, where $\epsilon\geq0$ and $\Omega$ is called the essential support of $V_{g,A}f(x,u)$,
we have
\[
|\Omega|\ge2\pi b(1-\epsilon)^{\frac{p}{p-2}}\Big(\frac{p}{2}\Big)^{\frac{2}{p-2}} \quad\textrm{ for all } p>2.
\]
\end{Corollary}

\begin{proof}
It follows from H\"{o}lder's inequality and (\ref{Thm:WOLCT:Lieb1}) in Theorem~\ref{Thm:WOLCT:Lieb} that
\begin{eqnarray*}
1-\epsilon
&\le&\int\int_{\Omega}|V_{g,A}f(x,u)|^2{\rm{d}}x{\rm{d}}u\\
&\le&\Big(\int\int_{\mathbb{R}^2}|V_{g,A}f(x,u)|^p{\rm{d}}x{\rm{d}}u\Big)^{2/p}|\Omega|^{\frac{p-2}{p}}\\
&\le&\big(2\pi b\big)^{2/p-1}\Big(\frac{2}{p}\Big)^{2/p}\|f\|_2^2\|g\|_2^2|\Omega|^{\frac{p-2}{p}}\\
&=&\big(2\pi b\big)^{2/p-1}\Big(\frac{2}{p}\Big)^{2/p}|\Omega|^{\frac{p-2}{p}}.
\end{eqnarray*}
Hence, for all $p>2$, we have
\[
|\Omega|\ge2\pi b(1-\epsilon)^{\frac{p}{p-2}}\Big(\frac{p}{2}\Big)^{\frac{2}{p-2}},
\]
which completes the proof.
\end{proof}

\section{Conclusion}\label{sec:O4}

In this paper, we first propose Donoho-Stark's uncertainty principle and Amrein-Berthier-Benedicks's uncertainty principle for the OLCT, which are different from Heisenberg's uncertainty principle for the OLCT obtained in \cite{Stern2007}. We then introduce the short-time OLCT and present its Lieb's uncertainty principle. Finally, we give a lower bound for the essential support of the short-time OLCT.



\end{document}